\documentclass[pdflatex,sn-mathphys-num]{sn-jnl}

\usepackage{graphicx}%
\usepackage{multirow}%
\usepackage{amsmath,amssymb,amsfonts}%
\usepackage{amsthm}%
\usepackage{mathrsfs}%
\usepackage[title]{appendix}%
\usepackage{xcolor}%
\usepackage{textcomp}%
\usepackage{manyfoot}%
\usepackage{booktabs}%
\usepackage{tikz}
\usepackage{algorithm}%
\usepackage{algorithmicx}%
\usepackage{algpseudocode}%
\usepackage{listings}%

\usepackage{todonotes}

\theoremstyle{thmstyleone}%
\newtheorem{theorem}{Theorem}
\newtheorem{proposition}[theorem]{Proposition}%

\theoremstyle{thmstyletwo}%
\newtheorem{example}{Example}%
\newtheorem{remark}{Remark}%
\newtheorem{lemma}{Lemma}

\theoremstyle{thmstylethree}%
\newtheorem{definition}{Definition}%

\begin{document}

\title[Photon surfaces extensions for dynamical gravitational collapse]{Photon surfaces extensions for dynamical gravitational collapse}


\author*[1,2]{\fnm{Roberto} \sur{Giamb\`o}}\email{roberto.giambo@unicam.it}

\author[1]{\fnm{Camilla} \sur{Lucamarini}}\email{camilla.lucamarini@studenti.unicam.it}
\equalcont{These authors contributed equally to this work.}


\affil*[1]{\orgdiv{School of Science and Technology, Mathematics Division}, \orgname{Università di Camerino}, \orgaddress{\street{Via Madonna delle Carceri 8}, \city{Camerino}, \postcode{62032}, \state{MC}, \country{Italy}}}

\affil[2]{\orgdiv{Osservatorio Astronomico di Brera}, \orgname{INAF}, \orgaddress{\street{Via Brera 28}, \city{Milano}, \postcode{20121},  \country{Italy}}}





\abstract{The equations for the photon surface in spherical symmetry are worked 
out, starting from \cite{Claudel:2000yi}, in the most general dynamical setting. 
We show that the condition for a timelike hypersurface to be a photon surface 
can be reformulated as a non-autonomous dynamical system, whose analysis reveals 
that the same condition also holds when the surface is generated by a null 
radial geodesic. 

As an application, we consider a well-known model of a spherical dust cloud 
undergoing gravitational collapse. In the marginally bound case the photon 
surface condition reduces to an exact planar dynamical system, whose analysis 
shows that the extension of the exterior photon surface $r=3M$ into the cloud is 
unique, and null, if and only if the central singularity is covered, an open set 
of timelike extensions being available otherwise. Comparing our findings with 
those in \cite{Cao:2019vlu}, we then investigate analytically whether the null extension covers the singularity developing in the Lemaitre-Tolman-Bondi (LTB) 
model.
}

\keywords{photon surface, photon sphere, uniqueness, dynamical systems, dust 
collapse, black hole, naked singularity}



\maketitle

\section{Introduction}\label{sec:intro}
Detection of black holes represents an intriguing challenge in the realm of observational astrophysics. The well-known international collaboration Event Horizon Telescope (EHT) has utilized a network of radio telescopes across the globe to elaborate an image of the black hole in galaxy M87 \cite{EventHorizonTelescope:2019dse,EventHorizonTelescope:2024dhe}, often referred to as the black hole shadow. The EHT collaboration has also applied analogous methods to investigate other substantial cosmic entities, notably the SgrA black hole situated at the Milky Way's core \cite{EventHorizonTelescope:2022wkp}.

It is important to note that the EHT images are not direct optical photographs but rather reconstructed radio maps that trace the emission from hot gas orbiting the black hole, which gives rise to the observed luminous ring.


Indeed, despite its name, the EHT does not aim to capture the shape of the event horizon itself, but rather to trace the motion of matter orbiting the black hole within the observable region to a faraway observer, which can be effectively described, in the simplest toy models, by the so-called \textit{photon sphere}.

Classically, the concept of a photon sphere arises in static and spherically symmetric settings such as the Schwarzschild solution, where it manifests as a timelike hypersurface at radius \( r = 3M \) on which initially tangent null geodesics remain tangent as they evolve. 
A similar approach, based on the study of null radial geodesics, has been used in the context of static spacetimes \cite{Nolan:2014maa,Hasse:2001by,Chakraborty:2011uj,Bogush:2023ojz,Bogush:2024fqj,Sadeghi:2024itx,Vertogradov:2024dpa}, or even in the context of non--transparent collapse \cite{Schneider:2018hge}.

When it comes to a fully dynamical setting, the situation is more involved. One could think of generalizing the above approach to this situation, which actually proves fruitful when the metric assumes a special form - for example, \cite{Mishra:2019trb,Solanki:2022glc,Vertogradov:2024eim} study
the situation for a spherical metric written in Bondi coordinates $(v,r,\theta,\phi)$, where $v$ is a null coordinate, as is the case for the Vaidya solution. 

However, this notion has been extended and generalized to a  broader geometric object known as the \textit{photon surface}, formalized in the seminal work by Claudel, Virbhadra, and Ellis \cite{Claudel:2000yi}, where it is defined independently of any specific symmetry or Killing structure.

In their framework, photon surfaces are understood as nowhere-spacelike hypersurfaces invariant under the flow of null geodesics that remain confined within the surface. This perspective allows for a more flexible analysis, applicable not only in static spacetimes \cite{Virbhadra:2002ju} but also in dynamical regimes where symmetry assumptions are weakened or absent. Subsequent studies, such as those by Cao and Song \cite{Cao:2019vlu}, have proposed quasi-local characterizations of photon surfaces in general spherically symmetric spacetimes, aiming to preserve the spirit of Claudel et al.'s geometric definition while adapting it to the complexity of evolving gravitational systems.

One physically rich setting in which photon surfaces can be investigated dynamically is gravitational collapse. The Lemaitre–Tolman–Bondi (LTB) models, describing inhomogeneous spherical dust clouds undergoing collapse under their own gravity, provide a valuable laboratory for probing how such surfaces behave in the presence of singularities and horizons. These models allow for analytical control over the metric and matter profiles and, depending on the initial data, can yield either black hole formation or naked singularities, offering insight into the end-state of light rays and causal curves near regions of extreme curvature \cite{Joshi:2008zz}.
Recent investigations have demonstrated that shadows and photon spheres, traditionally associated with black holes, can also arise in certain naked singularity spacetimes resulting from gravitational collapse, depending on specific parameters of the model \cite{Shaikh:2018lcc}. This challenges the conventional interpretation of shadow observations as definitive evidence for black holes.

The central aim of this work is to explore how photon surfaces, originally defined in static contexts, can be meaningfully extended into dynamical spacetimes -- specifically, within the LTB collapse model. By deriving and analyzing the governing equations for photon surfaces in the most general spherically symmetric dynamical setting, we assess whether and how these surfaces can be propagated from the exterior Schwarzschild geometry into the interior region occupied by the collapsing matter. In particular, we focus on whether such an extension is unique, and on whether 
it can effectively cover the singularity forming at the center, depending on the 
visibility of that singularity.

To carry out this investigation, we start by reviewing the photon surface condition for timelike hypersurfaces in spherically symmetric spacetimes, as given in \cite{Claudel:2000yi}, and derive the relevant differential equations that such surfaces must satisfy. We then apply this framework to the LTB model, examining how the photon surface extends across the boundary between interior and exterior spacetimes and determining the precise conditions under which such an extension exists and remains physically meaningful. In this respect, the present paper aims to compare and discuss its findings with those presented in \cite{Cao:2019vlu}, trying to highlight some aspects that might otherwise be overlooked, particularly regarding the relation between the photon surface behavior and the appearance of naked singularities.

Our main result in this direction presents a clear dichotomy. Since in the 
marginally bound case the metric coefficient is known explicitly, the photon 
surface condition can be recast, with no approximation involved, as a planar 
dynamical system. Its analysis shows that the extension of the exterior photon 
surface into the cloud is unique -- and then null -- if and only if the central 
singularity is covered; when it is naked, an open set of initial data gives rise 
to timelike extensions as well, and uniqueness fails. The threshold, 
$2m/r=4-2\sqrt3$, coincides exactly with the known critical value separating 
black hole and naked singularity formation 
\cite{Mena:2001dr,Giambo:2002tp}: uniqueness of the photon surface extension is 
itself a signature of the black hole end state.


Recent studies have shown that black holes and naked singularities in this 
context produce almost indistinguishable shadows \cite{Kong:2013daa,Kong:2013aja}. However, such similarity arises only at late times, whereas the underlying dynamics responsible for the shadows differ between the naked singularity and the black hole configurations \cite{Ortiz:2013wza,Ortiz:2015rma}. In our work, we highlight that the extension of the photon surface behaves differently between black holes and naked singularities, and these differences match the distinct shadow formation dynamics observed in previous research.

Before proceeding, it is worth stressing that horizons and photon surfaces, although both tied to the causal structure of spacetime, address different physical questions.
In the context of gravitational collapse, an apparent horizon is the outermost marginally trapped 2-sphere: the surface on which outgoing null geodesics have vanishing expansion, so that light directed outward neither spreads nor converges \cite{PhysRevLett.14.57,Hawking:1973uf}. It marks the boundary of the trapped region, beyond which even outward-pointing photons are dragged inward. A photon surface, by contrast, is a hypersurface with the property that every null geodesic initially tangent to it remains tangent under propagation \cite{Claudel:2000yi}; it captures the locus where photons can stay confined to the surface. In Schwarzschild spacetime the photon sphere envelops the event horizon, located in a stronger gravitational regime. We shall see that the same hierarchy persists in a dynamical setting also, and we will highlight how the two surfaces evolve differently during collapse.

The paper is organized as follows. In Section \ref{sec:ss}, we recall the geometric setting for photon surfaces in spherically symmetric spacetimes and derive their evolution equations. 
Section \ref{sec:dust} provides an overview of the fundamental concepts related to spherical dust collapse and reviews the  key features of the singularity emerging at the end of the collapse. 
Section \ref{sec:dust-ps} focuses on determining the photon surface for the 
model at hand, showing that its extension into the cloud is unique precisely 
when the central singularity is covered, and that the nature of the singularity 
also dictates whether the photon surface terminates at the regular center or 
reaches the singularity itself.
Finally, Section \ref{sec:outro} summarizes the main findings and suggests possible directions for further investigation.

\section{Photon surfaces in spherical symmetry}\label{sec:ss}
It is a well-established fact that Schwarzschild spacetime exhibits $r=3m$ as a hypersurface with the property that initially tangent null geodesics remain tangent, known as \textit{photon sphere}. As we shall see, the paper \cite{Claudel:2000yi} suggests that photon spheres arise as particular cases of more general objects, photon surfaces.
\begin{definition}\label{def:Ellis}
A \textit{photon surface} of $(M,g)$ is an immersed, nowhere–spacelike hypersurface $S$ of $(M,g)$ such that $\forall p\in S$ and $\forall k\in T_pS$, $\exists\gamma:(-\epsilon,\epsilon)\to M$ such that $\gamma(0)=p,\,\dot\gamma(0)=k$, and $\gamma((-\epsilon,\epsilon))\subset S$.
\end{definition}

\begin{remark}\label{rem:null}
    A null hypersurface is trivially a photon surface.
\end{remark}

In the spherically symmetric case, \cite{Claudel:2000yi} states a condition for a timelike hypersurface to be a photon surface:

\begin{theorem}\cite[Theorem III.1]{Claudel:2000yi}\label{thm:Ellis}
    Let $(M,g)$ be a spherically symmetric spacetime and $S$ an $SO(3)$--invariant timelike hypersurface of $(M,g)$, $X$ be the $SO(3)$--invariant unit future--directed timelike tangent vector field along $S$ orthogonal to the $SO(3)$--invariant 2--spheres in $S$. Let $\mathcal T$ be one such $SO(3)$--invariant 2--sphere in $S$ and $\mathcal T_s$ be the $SO(3)$--invariant 2--sphere in $S$ at arc length $s$ from $\mathcal T$ along the integral curves of $X$. Then $S$ is a photon surface of $(M,g)$ iff 
    \begin{equation}\label{eq:Ellis}
    \frac{\mathrm d^2}{\mathrm ds^2}A_s=\frac{1}{4A_s}\left(\frac{\mathrm d}{\mathrm ds}A_s\right)^2+
    A_s\left(
    \frac13\Theta^2-G_{\alpha\beta}n^\alpha n^\beta\right)-4\pi,
    \end{equation}
    where $A_s$ is the area of $\mathcal T_s$ and $\Theta$ is the expansion of the normal unit $n^\alpha$ to $S$. 
\end{theorem}

Now, let us consider a spacetime $M$ with spherical symmetry, characterized by coordinate variables $(t, x, \theta, \phi)$. The metric is expressed as follows:
\begin{equation}\label{eq:ss}
g=-e^{2\nu}\,\mathrm{d}t^2+e^{2\lambda}\,\mathrm{d}x^2+r^2\mathrm{d}\Omega^2,
\end{equation}
where $\lambda$, $\nu$, and $r$ are functions of $(t,x)$ and $\mathrm{d}\Omega^2=\mathrm{d}\theta^2+\sin^2\theta\,\mathrm{d}\phi^2$. We shall proceed to apply Theorem \ref{thm:Ellis}. Since $S$ is timelike, we can parameterize it with coordinates $(t,x(t),\theta,\phi)$, with $x(t)$ to be determined from the equation \eqref{eq:Ellis}. In this situation, indeed, the area of $\mathcal T_s$ is given by $A_s=4\pi r(t,x(t))^2$, and \eqref{eq:Ellis} takes the form
\begin{multline}\label{eq:ODE2}
\ddot x(t)=\frac{\left(r_x-r\nu_x\right) e^{2(\nu- \lambda)}}{r}+\dot x(t) \left(\frac{r_t}{r}-2 \lambda_t+\nu_t\right)
    \\
    +\dot x(t)^2 \left(-\frac{r_x}{r}-\lambda_x+2 \nu_x\right)+\dot x(t)^3 \frac{\left(r \lambda_t-r_t\right) e^{2( \lambda-\nu)}}{r},
\end{multline}
where $\dot x(t)$ and $\ddot x(t)$ stand for derivatives of $x$ with respect to $t$ and, given a function $f(t,x)$, $f_t$ and $f_x$ denote its partial derivatives. 

For the purpose of the present paper, the second-order ODE \eqref{eq:ODE2} is not the best way to study the photon surface and our strategy here is to shift the perspective from a purely geometric PDE to a dynamical systems framework: this crucial step is formalized in the following Proposition. 
\begin{proposition}\label{thm:ps}
Let $S$ be a photon surface for the metric \eqref{eq:ss} parameterized by $(t,x(t),\theta,\phi)$.
Then, the second-order PDE \eqref{eq:ODE2} is strictly equivalent to the following first-order non-autonomous dynamical system:
\begin{subequations}
    \begin{align}
\dot x(t)&=e^{\nu-\lambda} y(t),\label{eq:ODE1a}\\   
\dot y(t)&=\left(1-y(t)^2\right)\frac{ e^{\nu-\lambda} \left(r_x-r\nu_x\right)+y(t) \left(r_t-r\lambda_t\right)}{r}.\label{eq:ODE1b}    
    \end{align}
\end{subequations}
\end{proposition}
We stress that the reformulation provided by the above Proposition  is not merely a formal manipulation. By recasting \eqref{eq:ODE2} into the first-order system \eqref{eq:ODE1a}--\eqref{eq:ODE1b}, we reveal the underlying non-autonomous dynamical nature of the problem. What's more, this structure is decisive to discuss the extension of the photon 
surface into the collapsing dust cloud, and in particular to establish when such 
an extension is unique, as we will see in Section \ref{sec:dust-ps}.

\begin{remark}\label{rem:system}
Inspection of \eqref{eq:ODE1a}--\eqref{eq:ODE1b} above leads us to conclude that if $y(t_0)=\pm 1$ for some $t_0$, then $y(t)=\pm 1$ throughout the solution. Let us notice that $\dot x=\pm e^{\nu-\lambda}$ is the equation of a radial null geodesic in this spacetime. Therefore, although we already know (see Remark \ref{rem:null}) that null hypersurfaces are trivially photon surfaces, we can see that equation \eqref{eq:Ellis} in Theorem \ref{thm:Ellis}, which in principle assumes the surface to be timelike, also holds for those null hypersurfaces $S$ generated by a radial null geodesic $(t,x(t)) $ in the 2--dimensional quotient submanifold $M/S^2$. 
We should remember this fact, which will prove helpful when discussing the initial conditions for equation \eqref{eq:ODE2} or, equivalently, the system \eqref{eq:ODE1a}--\eqref{eq:ODE1b}. As we will observe in the examples we explore later, these initial conditions are essentially determined by the boundary conditions applied to the metric being analyzed.
\end{remark}

\begin{example}\label{rem:static}
Let us consider the particular case where \eqref{eq:ss} is static, therefore admitting a $SO(3)\times\mathbb R$ group of isometries, where the $\mathbb R$ isometries are generated by a Killing vector field $K$, orthogonal to the $SO(3)$ orbits. In this context, the notion of \textit{photon sphere} is defined in \cite{Claudel:2000yi} as an $SO(3)\times\mathbb R$ invariant photon surface. In the static case, $\nu,\lambda$ and $r$ are functions of $x$ only, and the Killing vector field that generates the isometries of $\mathbb R$ is collinear to $\partial_t$. Therefore, a photon space is 
a photon surface with constant $x(t)$, which in view of Proposition \ref{thm:ps} is implicitly defined by those $x$ satisfying
\begin{equation}\label{eq:phsph}
\nu'(x)=\frac{r'(x)}{r(x)}.
\end{equation}
For instance, in the case of the anisotropic generalizations of de Sitter spacetimes \cite{Giambo:2002wr} given by \footnote{Please note that in this example we are using $r$ as a coordinate in place of $x$, in compliance with most of the literature on the subject.} \begin{equation}\label{eq:anisodS}
g=-\left(1-\frac{2m(r)}{r}\right)\,\mathrm dt^2+\left(1-\frac{2m(r)}{r}\right)^{-1}\,\mathrm dr^2+r^2\mathrm d\Omega^2,
\end{equation}
equation \eqref{eq:phsph} takes the well-known form (see \cite[(67)]{Claudel:2000yi})
\[
1-\frac{3m(r)}{r}=m'(r),
\]
which in the special Schwarzschild case $m$ constant is solved exactly by $r=3m$. 
\end{example}

\section{Spherical dust collapsing spacetimes}\label{sec:dust}

\subsection{The global model}
\label{sec:collapse}
In the following, we will first review the basic properties of spherical dust spacetimes together with the features related to the nature of the spacetime singularities arising from complete collapse. 

Let us consider a general spherical metric \eqref{eq:ss} where the energy--momentum tensor is given by $T=-\rho\,\mathrm dt\otimes\partial_t$ \cite{Singh:1994tb,Joshi:2008zz}. Using Einstein field equations, the metric reduces to the 
Lema\^{\i}tre--Tolman--Bondi (LTB) model
\begin{equation}
    \label{eq:LTB}
    g=-\mathrm dt^2+\frac{r_x^2}{1+k(x)}\,\mathrm dx^2+r^2\,\mathrm d\Omega^2,
\end{equation}
where the collapse evolution is dictated by the PDE 
\begin{equation}\label{eq:LTB-pde}
r_t(t,x)=-\sqrt{k(x)+\frac{2m(x)}{r(t,x)}}.
\end{equation}
The initial data $k(x)$ and $m(x)$ can be prescribed at some initial time $t_0$, i.e., $t_0=0$, along with the initial energy $\rho_0(x):=\rho(t_0,x)$. From \eqref{eq:LTB-pde} the identity $m(x)=\frac{r}{2}\left(1-g^{\alpha\beta}r_\alpha r_\beta\right)$ is immediately derived: the function $m(x)$ is the so--called {\it Misner-Sharp mass}, and its choice will also determine the initial condition $r(t_0,x)$, because integrating the Einstein field equation $m_x=4\pi\rho r^2 r_x$ along $t=t_0$ one obtains
$$
r(t_0,x)^3=\int_0^x\frac{3m'(\xi)}{4\pi\rho_0(\xi)}\,\mathrm d\xi.
$$
The metric \eqref{eq:LTB} describes only the interior region of the model, i.e., the dust star. In other words, this metric is defined for $x\in[0,x_b]$, where $x_b$ represents the boundary of the star. The timelike hypersurface $\Sigma=\{x=x_b\}$ must be matched with a suitable exterior solution that describes the matter distribution outside the spherical dust cloud. 
This will be achieved using the Darmois junction conditions, stating that the matching of two spacetimes across a timelike hypersurface $\Sigma$ is smooth if the first and the second fundamental forms of the two metrics in $\Sigma$ respectively coincide 
\cite{Darmois1927,Israel:1966rt,Mars:1993mj}. 
It is well known that for the exterior region to be Schwarzschild spacetime (i.e., spherical vacuum), the radial pressure, indicated by the $T^1_1$ component of the energy-momentum tensor, must be zero at $\Sigma$. Naturally, this applies to the model being analyzed here, where pressures uniformly vanish throughout the spacetime. According to the Einstein field equations, this condition ensures that the Misner--Sharp mass remains constant across the matching hypersurface, a necessary condition for matching \cite{Magli_1997}.

It is worth observing that the construction of the global model—an LTB interior matched to a Schwarzschild exterior across the timelike hypersurface $\Sigma$—has the structure of an initial-boundary value problem (see e.g. \cite{Brito:2017ijp}): the Darmois conditions play the role of boundary conditions that must be satisfied throughout the evolution, not only on the initial slice. In the present case, these conditions reduce to the requirement that the mass is continuous along $\Sigma$, together with the identification of the areal radius on $\Sigma$, and they are automatically propagated by the explicit solution of \eqref{eq:LTB-pde}.

\subsection{Final fate of the marginally bound dust collapse}\label{sec:endstate}
From now on, we will confine ourselves to the so-called \textit{marginally bound case}, which amounts to setting $k(x)=0$. On the one hand, this choice does not exclude any of the physical implications of the model because, as one can demonstrate \cite{Giambo:2002tp}, the spacetime causal structure is reduced to the study of a dynamical system depending on the choice of some parameters typically related to the Taylor expansion of the mass function. This is an instance where a symmetry‑reduced Einstein system exhibits qualitatively distinct regimes separated in parameter space; see \cite{Cotsakis:2023emv,Cotsakis2025}.
On the other hand, however, the marginally bound assumption allows the equation \eqref{eq:LTB-pde} to integrate explicitly and yield 
\begin{equation}\label{eq:r-mb}
r(t,x)=x(1-\mu(x) t)^{2/3},\qquad\mu(x)=\frac32\sqrt{\frac{2m(x)}{x^3}},
\end{equation}
where the initial energy-density profile has also been chosen so that $r(0,x)=x$. In this way, each timelike hypersurface with fixed $x=x_0$ physically represents the shell of dust cloud particles that, at the initial time, are located at a proper distance $x_0$ from the center of symmetry, and the function $t\mapsto r(t,x_0)$ will represent the dynamical evolution of this shell as time flows by. The shell collapses under its own gravitational influence and eventually becomes singular at comoving time $t_s(x_0)=\mu(x_0)^{-1}$, where the energy density $\rho$ diverges, indicating the presence of a spacetime singularity, which can be covered by the presence of the apparent horizon $t_h(x_0)$,  described by the equation $r(t_h(x),x)=2m(x)$, see Figure \ref{fig:1}.
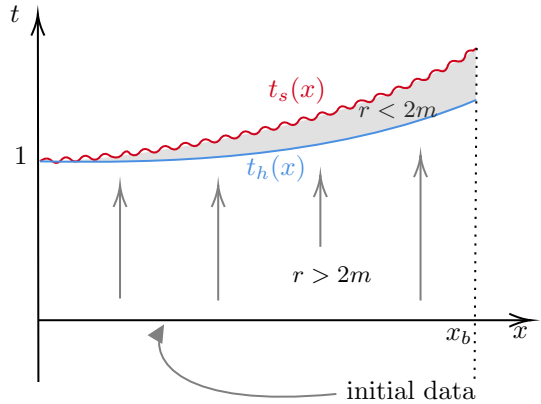
\begin{figure}[h]
  \centering

\tikzset{every picture/.style={line width=0.75pt}} 

\begin{tikzpicture}[x=0.75pt,y=0.75pt,yscale=-1,xscale=1]

\draw  [draw opacity=0][fill={rgb, 255:red, 155; green, 155; blue, 155 }  ,fill opacity=0.3 ] (172.29,114.36) .. controls (195.87,109.07) and (215.24,103) .. (230.67,97.16) .. controls (241.04,93.24) and (249.63,89.42) .. (256.54,86.02) .. controls (271.85,78.49) and (278.83,73.01) .. (278.41,73) .. controls (277.53,72.97) and (278.16,99.2) .. (278.22,99.11) .. controls (278.27,99.06) and (270.54,103.77) .. (248.96,110.15) .. controls (237.63,113.49) and (222.5,117.3) .. (202.67,121.11) .. controls (145.04,132.18) and (60.54,129.64) .. (60.41,130) .. controls (60.29,130.36) and (97.67,131.11) .. (172.29,114.36) -- cycle ;
\draw    (59.62,209.92) -- (303.62,209.92) ;
\draw [shift={(305.62,209.92)}, rotate = 180] [color={rgb, 255:red, 0; green, 0; blue, 0 }  ][line width=0.75]    (10.93,-3.29) .. controls (6.95,-1.4) and (3.31,-0.3) .. (0,0) .. controls (3.31,0.3) and (6.95,1.4) .. (10.93,3.29)   ;
\draw    (59.64,241) -- (59.37,57.8) ;
\draw [shift={(59.36,55.8)}, rotate = 89.91] [color={rgb, 255:red, 0; green, 0; blue, 0 }  ][line width=0.75]    (10.93,-3.29) .. controls (6.95,-1.4) and (3.31,-0.3) .. (0,0) .. controls (3.31,0.3) and (6.95,1.4) .. (10.93,3.29)   ;
\draw [color={rgb, 255:red, 208; green, 2; blue, 27 }  ,draw opacity=1 ]   (60.41,130) .. controls (62.66,128.22) and (64.61,128.12) .. (66.26,129.69) .. controls (67.94,131.26) and (69.52,131.16) .. (71,129.39) .. controls (72.47,127.62) and (74.06,127.5) .. (75.78,129.04) .. controls (77.52,130.57) and (79.13,130.44) .. (80.61,128.65) .. controls (82.08,126.84) and (83.7,126.7) .. (85.47,128.21) .. controls (87.26,129.71) and (88.9,129.55) .. (90.37,127.73) .. controls (91.83,125.9) and (93.47,125.73) .. (95.3,127.2) .. controls (97.97,128.58) and (100.03,128.34) .. (101.49,126.49) .. controls (102.94,124.63) and (104.19,124.48) .. (105.23,126.03) .. controls (107.11,127.47) and (108.78,127.25) .. (110.23,125.38) .. controls (111.67,123.5) and (113.34,123.27) .. (115.24,124.69) .. controls (117.15,126.1) and (118.83,125.85) .. (120.26,123.95) .. controls (121.68,122.05) and (123.35,121.79) .. (125.28,123.18) .. controls (127.23,124.56) and (128.9,124.29) .. (130.31,122.37) .. controls (131.71,120.45) and (133.39,120.17) .. (135.34,121.53) .. controls (137.31,122.88) and (138.98,122.58) .. (140.36,120.64) .. controls (141.73,118.69) and (143.4,118.39) .. (145.37,119.72) .. controls (147.36,121.04) and (149.02,120.72) .. (150.37,118.76) .. controls (151.71,116.79) and (153.37,116.46) .. (155.35,117.76) .. controls (157.34,119.05) and (159,118.71) .. (160.31,116.73) .. controls (161.6,114.75) and (163.25,114.4) .. (165.25,115.67) .. controls (167.25,116.93) and (168.88,116.56) .. (170.15,114.57) .. controls (171.4,112.57) and (173.02,112.19) .. (175.02,113.43) .. controls (177.02,114.67) and (178.63,114.28) .. (179.85,112.27) .. controls (181.04,110.26) and (182.64,109.86) .. (184.64,111.07) .. controls (186.64,112.28) and (188.22,111.87) .. (189.38,109.84) .. controls (190.52,107.81) and (192.09,107.39) .. (194.08,108.58) .. controls (196.07,109.76) and (197.62,109.33) .. (198.72,107.3) .. controls (199.79,105.26) and (201.31,104.82) .. (203.3,105.98) .. controls (205.29,107.13) and (207.17,106.56) .. (208.94,104.29) .. controls (209.93,102.24) and (211.4,101.77) .. (213.37,102.9) .. controls (215.35,104.02) and (216.81,103.55) .. (217.74,101.49) .. controls (218.65,99.44) and (220.43,98.84) .. (223.09,99.69) .. controls (225.05,100.78) and (226.45,100.29) .. (227.28,98.22) .. controls (228.09,96.16) and (229.79,95.54) .. (232.4,96.35) .. controls (234.34,97.4) and (235.67,96.89) .. (236.4,94.82) .. controls (237.73,92.5) and (239.35,91.85) .. (241.26,92.88) .. controls (243.81,93.63) and (245.38,92.97) .. (245.97,90.9) .. controls (247.13,88.57) and (248.65,87.9) .. (250.52,88.88) .. controls (253,89.57) and (254.46,88.89) .. (254.9,86.83) .. controls (255.85,84.51) and (257.53,83.68) .. (259.93,84.33) .. controls (261.78,85.24) and (263.11,84.54) .. (263.91,82.21) .. controls (264.61,79.91) and (266.12,79.05) .. (268.43,79.63) .. controls (270.74,80.18) and (272.15,79.3) .. (272.65,77.01) .. controls (273.05,74.74) and (274.36,73.86) .. (276.57,74.35) -- (278.41,73) ;
\draw [color={rgb, 255:red, 74; green, 144; blue, 226 }  ,draw opacity=1 ]   (60.41,130) .. controls (117.41,131) and (204.22,130.11) .. (278.22,99.11) ;
\draw  [dash pattern={on 0.84pt off 2.51pt}]  (278.41,73) -- (277.33,242.67) ;
\draw [color={rgb, 255:red, 128; green, 128; blue, 128 }  ,draw opacity=1 ]   (121.01,214.83) .. controls (115.09,230.23) and (130.37,254.52) .. (208.39,247) ;
\draw [shift={(122.39,212)}, rotate = 120.96] [fill={rgb, 255:red, 128; green, 128; blue, 128 }  ,fill opacity=1 ][line width=0.08]  [draw opacity=0] (8.93,-4.29) -- (0,0) -- (8.93,4.29) -- cycle    ;
\draw [color={rgb, 255:red, 128; green, 128; blue, 128 }  ,draw opacity=1 ]   (100.39,199) -- (100.39,144) ;
\draw [shift={(100.39,142)}, rotate = 90] [color={rgb, 255:red, 128; green, 128; blue, 128 }  ,draw opacity=1 ][line width=0.75]    (10.93,-3.29) .. controls (6.95,-1.4) and (3.31,-0.3) .. (0,0) .. controls (3.31,0.3) and (6.95,1.4) .. (10.93,3.29)   ;
\draw [color={rgb, 255:red, 128; green, 128; blue, 128 }  ,draw opacity=1 ]   (149.39,200) -- (149.39,144) ;
\draw [shift={(149.39,142)}, rotate = 90] [color={rgb, 255:red, 128; green, 128; blue, 128 }  ,draw opacity=1 ][line width=0.75]    (10.93,-3.29) .. controls (6.95,-1.4) and (3.31,-0.3) .. (0,0) .. controls (3.31,0.3) and (6.95,1.4) .. (10.93,3.29)   ;
\draw [color={rgb, 255:red, 128; green, 128; blue, 128 }  ,draw opacity=1 ]   (200.39,173) -- (200.39,139) ;
\draw [shift={(200.39,137)}, rotate = 90] [color={rgb, 255:red, 128; green, 128; blue, 128 }  ,draw opacity=1 ][line width=0.75]    (10.93,-3.29) .. controls (6.95,-1.4) and (3.31,-0.3) .. (0,0) .. controls (3.31,0.3) and (6.95,1.4) .. (10.93,3.29)   ;
\draw [color={rgb, 255:red, 128; green, 128; blue, 128 }  ,draw opacity=1 ]   (250.39,200) -- (250.39,129) ;
\draw [shift={(250.39,127)}, rotate = 90] [color={rgb, 255:red, 128; green, 128; blue, 128 }  ,draw opacity=1 ][line width=0.75]    (10.93,-3.29) .. controls (6.95,-1.4) and (3.31,-0.3) .. (0,0) .. controls (3.31,0.3) and (6.95,1.4) .. (10.93,3.29)   ;

\draw (295,211.4) node [anchor=north west][inner sep=0.75pt]    {$x$};
\draw (44,50.4) node [anchor=north west][inner sep=0.75pt]    {$t$};
\draw (188.81,102.71) node [anchor=south] [inner sep=0.75pt]  [color={rgb, 255:red, 208; green, 2; blue, 27 }  ,opacity=1 ]  {$t_{s}( x)$};
\draw (162,125.4) node [anchor=north west][inner sep=0.75pt]  [color={rgb, 255:red, 74; green, 144; blue, 226 }  ,opacity=1 ]  {$t_{h}( x)$};
\draw (277.53,212.4) node [anchor=north east] [inner sep=0.75pt]    {$x_{b}$};
\draw (46,121.4) node [anchor=north west][inner sep=0.75pt]    {$1$};
\draw (212,239) node [anchor=north west][inner sep=0.75pt]   [align=left] {initial data};
\draw (218,98.4) node [anchor=north west][inner sep=0.75pt]  [font=\scriptsize,color={rgb, 255:red, 0; green, 0; blue, 0 }  ,opacity=0.82 ]  {$r< 2m$};
\draw (185,181.4) node [anchor=north west][inner sep=0.75pt]  [font=\scriptsize]  {$r >2m$};

\end{tikzpicture}

  \caption{Schematic picture of the collapse of spherical dust. The process begins from a regular configuration at $t=0$. Every shell $x \in [0, x_b]$ within the spherical cloud undergoes collapse, becoming trapped at comoving time $t_h(x)$ and fully collapsing at $t_s(x)$. The symmetry center gets trapped simultaneously as it becomes singular, allowing for the chance that photons might escape from the central singularity and reach the region $r>2m$.}
  \label{fig:1}
\end{figure}

To sum up, the global model that we are going to consider will therefore be composed of the interior LTB metric \eqref{eq:LTB} with $k(x)=0$ and $r(t,x)$ given by \eqref{eq:r-mb}, defined in the set 
$$ \{(t,x)\,:\,x\in[0,x_b],\,\,t<t_s(x)\}, $$
matched to a Schwarzschild exterior 
\begin{equation}\label{eq:ext}
g_E=-\left(1-\frac{2M}{r}\right)\,d\tau^2+\frac{1}{1-\frac{2M}{r}}\,\mathrm dr^2+r^2\,\mathrm d\Omega^2
\end{equation}
at the star boundary $\Sigma=\{ x=x_b\}$. 
As discussed in the previous subsection,
the Darmois junction conditions ensure that the matching is $C^1$ across $\Sigma$, provided that the Schwarzschild mass  is given by $M=m(x_b)$. Observe that the coordinates $(\tau, r)$ in \eqref{eq:ext} will be related to $t$ on the matching surface -- in particular, $r=r(t,x_b)$ on $\Sigma$\footnote{That justifies the slight abuse of notation where we employ the same letter $r$ for both the radial coordinate in \eqref{eq:ext} and the metric coefficient \eqref{eq:r-mb} within \eqref{eq:LTB}.}. Details can be found in \cite{Giambo2005GravitationalFields}.

We will refer to this model as \textit{the marginally bound spherical dust collapse}.
Its end state, starting from regular initial conditions at $t=0$, is well known in the literature \cite{Singh:1994tb}: let us summarize it briefly. Consider a mass profile such that, using \eqref{eq:r-mb}
\begin{equation}\label{eq:mu}
\mu(x)=1-a x^n + o(x^n),
\end{equation}
with $a>0$ and $n\in\mathbb{N}\setminus\{0\}$, so the initial energy $\rho_0(x)$ is regular up to the central shell $x=0$, and decreasing in a right neighborhood of $x=0$ -- we can suppose that the star boundary $x_b$ belongs to this neighborhood. In this configuration, the non-central singularities, i.e., $t_s(x)$ for every $x\in]0,x_b]$ are always covered by the apparent horizon. Conversely, the central singularity $t_s(0)$ may be either covered or naked, depending on the choice of parameters $a$ and $n$ in \eqref{eq:mu}. Indeed, one has to observe that deriving $t_h(x)$ from the condition $r=2m$ leads to the expression
$$
t_h(x)=\frac{1}{\mu(x)}\left(1-\frac{8}{27}\mu(x)^3 x^3\right),
$$
implying $t_h(0)=t_s(0)=1$, see again Figure \ref{fig:1}. 

Full details of these studies are based on the qualitative study of radial null geodesics and can be found in \cite{Mena:2001dr,Giambo:2002tp}. According to these findings, the central singularity is:
\begin{itemize}
    \item naked if $n=1,2$ or $n=3$ and $a>\frac{2(26+15\sqrt 3)}{27}$, and
    \item covered, otherwise. 
\end{itemize}

For later use, let us record two properties of the outgoing radial null 
geodesics of the model. Writing such a geodesic as $t=t(x)$, it solves the ODE
\begin{equation}\label{eq:nullgeo}
\frac{\mathrm dt}{\mathrm dx}=r_x(t(x),x),
\end{equation}
which is regular as long as $x>0$ and $t<t_s(x)$, because $r_x>0$ there.

The first property, established in \cite{Giambo:2001wi}, is that the apparent 
horizon curve is a \emph{subsolution} of \eqref{eq:nullgeo}, i.e.
$
t_h'(x)<r_x(t_h(x),x)$, with $x\in]0,x_b].$
In the marginally bound case this can be checked at once, since a direct 
computation using \eqref{eq:r-mb} gives the explicit identity
\begin{equation}\label{eq:subsol}
r_x(t_h(x),x)-t_h'(x)=\frac49\,\mu(x)\,x^2\bigl(3\mu(x)-2x|\mu'(x)|\bigr),
\end{equation}
whose right hand side is positive in a right neighborhood of $x=0$, and we may 
assume that the star boundary $x_b$ belongs to this neighborhood. As a 
consequence, a solution of \eqref{eq:nullgeo} lying below the apparent horizon 
at some $\bar x$ stays below it for every $x<\bar x$: the region 
$\{t<t_h(x)\}$ is invariant under backward integration, and in particular such 
a solution never meets the singularity curve $t_s(x)$ at $x>0$.

The second property is that 
the singularity is naked exactly when 
\eqref{eq:nullgeo} admits solutions, lying below $t_h(x)$, with $t(x)\to t_s(0)=1$ as $x\to0^+$, that 
is, null geodesics emanating from the central singularity. Accordingly, when 
the central singularity is covered, every solution of \eqref{eq:nullgeo} lying 
below the apparent horizon extends backward up to $x=0$ and reaches the 
\emph{regular} centre,
\begin{equation}\label{eq:reg-centre}
\lim_{x\to0^+}t(x)=t_*<t_s(0)=1 .
\end{equation}

\section{Photon surfaces in dust collapse}\label{sec:dust-ps}

\subsection{Extension to the interior LTB metric}\label{sec:ps-ext}
Let us now determine the photon surface for the marginally bound dust collapse. On the exterior \eqref{eq:ext}, of course, the photon surface is given by $r=3M$, see Example \ref{rem:static}. 

For the LTB interior metric,
as pointed out in \cite[(5.5)]{Cao:2019vlu}, a lengthy but straightforward computation transforms the equation \eqref{eq:ODE2} for photon surfaces as
\begin{multline}
    \label{eq:ps-dust}
    \ddot x(t)=
    \frac{k(x)+1}{r r_x}+
     \left(\frac{r_t}{r}-\frac{2 r_{tx}}{r_x}\right)\dot x(t)\\
-    \left(\frac{r_x}{r}+\frac{r_{xx}}{r_x}-\frac{k'(x)}{2 (k(x)+1)}\right)\dot x(t)^2 +\frac{r_x \left(r\, r_{tx}-r_x r_t\right)}{r(k(x)+1)}\dot x(t)^3, 
\end{multline}
that applying Proposition \ref{thm:ps} can be given in the form of the following 1st-order system:
\begin{subequations}
    \begin{align}
    \dot x(t)&=\frac{\sqrt{k(x)+1}}{r_x}\,y(t),\label{eq:ODE1a-mb}\\
    \dot y(t)&=(1-y(t)^2)\frac{r_x\sqrt{k(x)+1} +y(t) \left(r_x r_t-r\, r_{tx}\right)}{r\,r_x}.\label{eq:ODE1b-mb}
\end{align}
\end{subequations}
The crucial step now is to set initial conditions for the system \eqref{eq:ODE1a-mb}--\eqref{eq:ODE1b-mb}. The first condition is necessarily given by $x(t_0)=x_b$, where $t_0$ satisfies the equation
$r(t_0,x_b)=3m(x_b)$, i.e., using \eqref{eq:r-mb},
\begin{equation}\label{eq:t0}
t_0=\frac{1}{\mu(x_b)}\left(1-(2/3)^{3/2}\mu(x_b)^3x_b^3\right)=
\frac{\sqrt{2}}{3}  \left(\sqrt{\frac{x_b^3}{m(x_b)}}-3 \sqrt{3} m(x_b)\right)
\end{equation}
Consequently, the second condition must be expressed as $y(t_0) = y_0$. In the 
following, we are going to show that the admissible choices of $y_0$ are 
dictated by the nature of the central singularity: $y_0=1$ is the only 
physically acceptable option when the singularity is covered, whereas an open 
set of values $y_0<1$ becomes admissible when it is naked. Let's see why.

First of all, in order for the extension not to be spacelike at $t_0$, we need $|y_0|\le 1$, as we can see using \eqref{eq:ODE1a-mb} and \eqref{eq:LTB}. 

In view of Remark \ref{rem:null}, we can accept $y_0=\pm 1$ as initial conditions, since they will correspond to the outgoing ($y_0=1$) and ingoing ($y_0=-1$) radial null geodesics starting at $x=x_b$. Among the above cases, the outgoing curve will necessarily cover the ingoing one, which makes the $y_0=1$ extension more significant than the $y_0=-1$ one.

Of course, one may consider other extensions that start spacelike at $x=x_b$. These choices correspond to setting $|y_0|<1$. However, $y_0< 0$ would correspond to a decreasing $x(t)$ curve and would, for the same reason, be less significant than the $y_0\ge 0$ extension.

Before proving that, it is convenient to rewrite the system 
\eqref{eq:ODE1a-mb}--\eqref{eq:ODE1b-mb} in a form which highlights its dynamical 
content, and which will be used both here and in the analysis of 
the naked case. As long as $y>0$ the function $x(t)$ is strictly increasing, so 
it can be inverted and $x$ can be taken as the independent variable. It is then 
natural to introduce
\begin{equation}\label{eq:uchiLam}
u:=1-\mu(x)\,t,\qquad 
\chi:=-r_t=\frac23\,\frac{x\,\mu(x)}{u^{1/3}},\qquad 
\Lambda:=\frac{x\,r_x}{r}-1=\frac23\,\frac{t\,x\,|\mu'(x)|}{u},
\end{equation}
so that $r=x\,u^{2/3}$ and $\Lambda\ge0$ whenever $t\ge0$. 
The quantity $\chi$ 
has a direct geometric meaning, since \eqref{eq:LTB-pde} with $k\equiv0$ gives
\begin{equation}\label{eq:chi2}
\chi^2=\frac{2m(x)}{r},
\end{equation}
so that $\chi<1$ characterizes the region outside the apparent horizon, while 
$\Lambda$ measures the deviation of $r_x$ from its value at a homogeneous 
profile. Notice that, at fixed $x>0$, we have $\partial\chi/\partial t>0$: therefore,
$\chi$ may 
be used in place of $t$ as a dependent variable, with $x$ playing the role of 
the independent one.

Setting $s:=\ln x$, so that the backward extension towards the centre 
corresponds to $s\to-\infty$, and $w:=\operatorname{arctanh}y\in\,]0,+\infty]$, 
a direct computation turns 
\eqref{eq:ODE1a-mb}--\eqref{eq:ODE1b-mb} into the planar system
\begin{subequations}
\begin{align}
\frac{\mathrm d\chi}{\mathrm ds}&=
\chi\left[1+\frac{x\mu'(x)}{\mu(x)}-\frac\Lambda2
+\frac{\chi\,(1+\Lambda)}{2\,y}\right],\label{eq:chieq}\\
\frac{\mathrm dw}{\mathrm ds}&=\frac{1+\Lambda}{y}
+\frac23\,x^2\mu'(x)\,u^{-4/3}
\;=\;\frac{1+\Lambda}{y}-\frac32\,\frac{\Lambda\,\chi}{\mu(x)\,t}.
\label{eq:weq}
\end{align}
\end{subequations}
The 
line $y=1$, i.e.\ $w=+\infty$, is invariant, and corresponds to the null radial 
geodesics.

Finally, eliminating $u$ between the definitions in \eqref{eq:uchiLam} one 
finds that the two functions are not independent, being related by
\begin{equation}\label{eq:LamB}
\Lambda=B(x,t)\,\chi^3,\qquad 
B(x,t):=\frac94\,\frac{t\,|\mu'(x)|}{x^2\,\mu(x)^3}
=\frac94\,n\,a\,x^{\,n-3}\bigl(1+o(1)\bigr)\quad\text{as }x\to0^+,
\end{equation}
the last expression following from \eqref{eq:mu} along trajectories with 
$t\to1$. The exponent $n-3$ already suggests that the value $n=3$ plays a 
distinguished role, and that the two terms on the right hand side of 
\eqref{eq:weq} compete only when $n\le3$, that is, precisely in the range in 
which the central singularity may be naked.

\begin{lemma}\label{lem:ext}
    Assume that the central singularity is covered. Let $(x(t),y(t))$ be a 
    solution of \eqref{eq:ODE1a-mb}--\eqref{eq:ODE1b-mb} such that 
    $(x(t_0),y(t_0))=(x_b,y_0)$ with $y_0\in]0,1[.$ Then 
    $\exists\tilde t<t_0\,:\,x(\tilde t)=x_b.$ 
\end{lemma}

\begin{proof}
Before presenting the proof, we illustrate its core concept in Figure 
\ref{fig:V}: the (non--autonomous) vector field induced by 
\eqref{eq:ODE1a-mb}--\eqref{eq:ODE1b-mb} prevents the existence of curves with 
initial data $|y_0|<1$ that extend back to $x=0$.

\begin{figure}[h]
  \centering
  \includegraphics[width=0.8\textwidth]{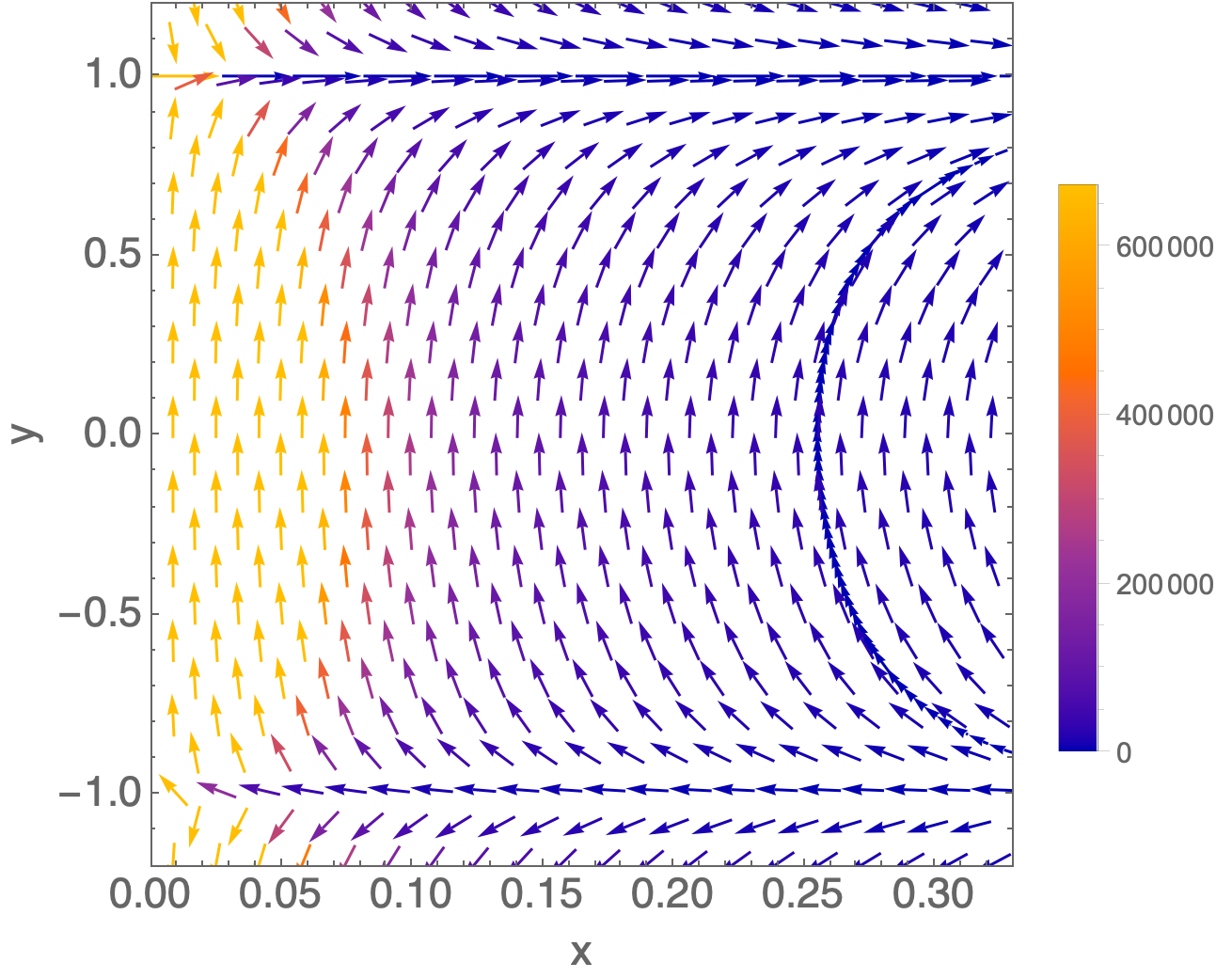}
\caption{The behavior of the vector field associated to 
\eqref{eq:ODE1a-mb}--\eqref{eq:ODE1b-mb}, in the case of a covered central 
singularity, is such that the only way to left--extend a solution from $x=x_b$ 
to $x=0$ is when $y_0=1$. Notice that the field is non autonomous (here is 
represented the situation at $t=1.01$ only). Other parameters are 
$n=4,\,a=8$.}
  \label{fig:V}
\end{figure}

After the above intuitive idea, let us start the proof by showing, as an 
intermediate step, that $\exists t_1<t_0\,:\,y(t_1)=0$. By contradiction, 
suppose that $y(t)>0,\ \forall t<t_0$. Then $x(t)$ is increasing 
$\forall t<t_0$ and then, called $t_{\text{inf}}\in\mathbb R\cup\{-\infty\}$ 
the infimum of the maximal interval where the solution can be extended, there 
exists $\lim_{t\to t_{\text{inf}}^+}x(t)=x_0\in[0,x_b[$. 

Suppose $t_{\text{inf}}=-\infty$. Using \eqref{eq:r-mb} and \eqref{eq:mu} in 
\eqref{eq:ODE1a-mb}--\eqref{eq:ODE1b-mb} we get 
\begin{equation}\label{eq:ext1}
\frac{\dot y(t)}{1-y(t)^2}\sim \frac{1}{x(1-a x^n)^{2/3}}\cdot\frac1{t^{2/3}},
\quad\text{\ as\ } t\to-\infty,
\end{equation}
which implies the existence of positive constants $C_1$, $C_2$ such that
$$
\frac{1-y(t)}{1+y(t)}\ge e^{-C_1 t^{1/3}+C_2}.
$$
Since the right-hand side positively diverges as $t\to-\infty$, we get a 
contradiction, recalling that $y(t)$ should remain positive. 

Then $t_{\text{inf}}\in\mathbb R$, and this situation may occur only when there 
is a loss of regularity in the system, that is, when $r\,r_x$ vanishes at the 
limit point. Since $r_x>0$ in the domain of the model, this happens only where 
$r=0$, i.e.\ either $x_0=0$, or $x_0>0$ and 
$t_{\text{inf}}=t_s(x_0)$. The latter case is excluded: since $y\le1$, along 
the trajectory
\begin{equation}\label{eq:supersol}
\frac{\mathrm dt}{\mathrm dx}=\frac{r_x(t(x),x)}{y}\ge r_x(t(x),x),
\end{equation}
so that $t(x)$ is a supersolution of \eqref{eq:nullgeo}; recalling that 
$t_0<t_h(x_b)$ -- which follows from $r(t_0,x_b)=3m(x_b)>2m(x_b)=r(t_h(x_b),x_b)$ 
and $r_t<0$ -- and that $t_h$ is a subsolution by \eqref{eq:subsol}, the same 
comparison argument used there gives $t(x)<t_h(x)<t_s(x)$ for every $x<x_b$.
Hence $x_0=0$.

It remains to exclude that the trajectory reaches $x=0$. Here the assumption on 
the central singularity enters. Let $t_N$ be the solution of 
\eqref{eq:nullgeo} with $t_N(x_b)=t_0$: integrating \eqref{eq:supersol} 
backwards from $x_b$ we get $t(x)\le t_N(x)$ for every $x\in]0,x_b]$, and since 
the singularity is covered, \eqref{eq:reg-centre} gives 
$t_N(x)\to t_*<1$ as $x\to0^+$. Consequently
$$
u(x)=1-\mu(x)t(x)\ \ge\ 1-\mu(x)t_N(x)\ \xrightarrow[x\to0^+]{}\ 1-t_*>0,
$$
so that there exists $\delta>0$ with $u(x)\ge\delta$ on $]0,x_b]$. Since $t(x)$ 
is bounded and $\mu'(x)=O(x^{n-1})$, the definitions \eqref{eq:uchiLam} then 
give $\Lambda=O(x^{n})$, while the second term in \eqref{eq:weq} is 
$O(x^{n+1})$. Therefore, using $0<y\le1$, equation \eqref{eq:weq} yields a 
positive constant $C_3$ such that
$$
\frac{\mathrm dw}{\mathrm ds}\ \ge\ 1-C_3\,x^{n}\ \ge\ \frac12
$$
for $x_b$ small enough, whence, integrating backwards from $x_b$,
\begin{equation}\label{eq:ext2}
w(x)\ \le\ w(x_b)-\frac12\,\ln\frac{x_b}{x}\ \xrightarrow[x\to0^+]{}\ -\infty,
\end{equation}
which contradicts $w>0$. 

Then so far we have shown that $y(t_1)=0$ for some $t_1<t_0$. Inspection of the 
equivalent equation \eqref{eq:ps-dust} shows that, if $\dot x=0$, then 
$\ddot x>0$, so $x(t_1)$ is a local minimum for $x(t)$. Hence, in a left 
neighborhood of $t_1$, $x(t)$ is (strictly) decreasing and therefore $y(t)$ is 
negative. This fact forces  $y(t)<0,\,\forall t<t_1$ and then $x(t)$ will be 
decreasing for earlier times than $t_1$. Let us show that $x(\bar t)=x_b$ for 
some $\bar t<t_1$. If that is not the case, then $(x(t),y(t))$ will remain in 
an area where the ODE is regular ($x(t)>x(t_1)$ and $y\in]-1,0]$). This means 
that the solution can be extended $\forall t<t_1$. But in this case the same 
estimate as \eqref{eq:ext1} gives $y(t)\to -1$, and using \eqref{eq:ODE1a-mb} 
we get the existence of a positive constant $C_4$ such that 
\begin{equation}\label{eq:ext3}
\dot x\sim-\frac{C_4}{t^{2/3}},
\end{equation}
that integrated as $t\to -\infty$ gives $x(t)$ positively diverging, which is a 
contradiction. This concludes the proof of the lemma.
\end{proof}

The above Lemma singles out extensions of photon surfaces that are not 
physically acceptable. The reason is that these extensions have two points of 
contact with the star boundary -- one at $t_0$ and the other at an earlier 
comoving time $\bar t<t_0$. Since we are considering a collapsing model, i.e.\ 
$r_t<0$, we have $r(\bar t,x_b)>r(t_0,x_b)=3M$. Therefore, the extension at 
$\bar t$ cannot be matched with the exterior photon surface, which is indeed 
matched at $t=t_0$. When the central singularity is covered, we may therefore 
conclude that the null extension is the only possible one.

The situation changes when the central singularity is naked, and it does so for 
a structural reason. The proof of Lemma \ref{lem:ext} rests on the positivity of 
the right hand side of \eqref{eq:weq}, forcing $w$ to decrease backward until 
$y$ vanishes. By \eqref{eq:LamB}, the negative term in \eqref{eq:weq} is of 
order $\Lambda=B\chi^3$ with $B\simeq\frac94na\,x^{\,n-3}$, which competes with 
the positive one only when $n\le3$ -- exactly the range in which the 
singularity may fail to be covered. There the two terms may balance, $w$ may 
increase backward instead, and the trajectory can be driven towards the central 
singularity with $y\to1$, never returning to $\Sigma$. This is shown below by 
exhibiting an explicit invariant region for the system 
\eqref{eq:chieq}--\eqref{eq:weq}.
To this aim, let us introduce the two functions
\begin{equation}\label{eq:gh}
g(\chi):=\frac{2+\chi}{\chi^3(1-\chi)},\qquad 
h(\chi):=\frac{2}{\chi^3(3\chi-2)},\qquad \chi\in\,]2/3,1[,
\end{equation}
which, as we are going to see, control respectively the sign of 
\eqref{eq:chieq} and of \eqref{eq:weq} on the invariant line $y=1$. A direct 
computation shows that $g(\chi)=h(\chi)$ if and only if $\chi^2+2\chi-2=0$, that 
is, at the single value
\begin{equation}\label{eq:chic}
\chi_{\mathrm c}=\sqrt3-1,\qquad 
g(\chi_{\mathrm c})=h(\chi_{\mathrm c})=\frac{26+15\sqrt3}{2}.
\end{equation}
It is remarkable that this threshold is not a new one: recalling 
\eqref{eq:LamB}, for $n=3$ one has $B\to\frac{27}{4}a$ as $x\to0^+$, so that the 
inequality $B>g(\chi_{\mathrm c})$ reads
$$
a>\frac{2\,(26+15\sqrt3)}{27},
$$
which is precisely the condition, recalled in Section \ref{sec:endstate}, for 
the central singularity to be naked.

\begin{proposition}\label{prop:naked}
Assume that the central singularity is naked, i.e.\ $n=1,2$, or $n=3$ and 
$a>\frac{2(26+15\sqrt3)}{27}$. Then, for $x_b$ small enough, there exists 
$y_*\in\,]0,1[$ such that, for every $y_0\in[y_*,1[$, the solution of 
\eqref{eq:ODE1a-mb}--\eqref{eq:ODE1b-mb} with initial data $(x_b,y_0)$ is 
defined on the maximal backward interval $]t_{\text{inf}},t_0]$, with 
$t_{\text{inf}}=t_s(0)=1$, and satisfies there
$$
y(t)>0,\qquad \frac{2m}{r}\ge 4-2\sqrt3,\qquad 
\lim_{t\to t_{\text{inf}}^+}x(t)=0 .
$$
In particular, the trajectory never returns to $\Sigma$, contrary to what 
happens in the covered case (Lemma \ref{lem:ext}), and the corresponding 
hypersurface is a \emph{timelike} photon surface extending backwards from $\Sigma$ up to the central singularity.
\end{proposition}

\begin{proof}
Throughout the proof we take $x$ as the independent variable and $(\chi,y)$ as 
the dependent ones, which is legitimate as said after \eqref{eq:chi2}. Consider, for $y_*\in\,]0,1[$ to be chosen, the planar 
region
$$
\mathcal R:=\left\{(\chi,y)\in\mathbb R^2\,:\,\chi_{\mathrm c}\le\chi<1,\ 
y_*\le y\le 1\right\},
$$
and let us show that the trajectory $x\mapsto(\chi(x),y(x))$ remains in 
$\mathcal R$ for every $x\in\,]0,x_b]$, that is, along the backward evolution 
$s\to-\infty$.

First, the initial datum belongs to $\mathcal R$. Indeed 
$r(t_0,x_b)=3m(x_b)$, so that \eqref{eq:chi2} gives 
$\chi(t_0,x_b)=\sqrt{2/3}\in\,]\chi_{\mathrm c},1[$, a value which is exact and 
independent of $x_b$ and of the mass profile, while $y_0\ge y_*$ by assumption.

Two sides of $\partial\mathcal R$ require no computation. The side $y=1$ is an 
invariant line for \eqref{eq:chieq}--\eqref{eq:weq}, and cannot be crossed. The 
side $\chi=1$ is never reached: since $y\le1$, the trajectory is a supersolution 
of \eqref{eq:nullgeo}, as in \eqref{eq:supersol}, and the comparison with the 
subsolution \eqref{eq:subsol} -- an argument which does not depend on the nature 
of the central singularity -- gives $t(x)<t_h(x)$, that is, $\chi<1$ throughout 
by \eqref{eq:chi2}. In particular $u>0$, so that the system is regular as long 
as $x>0$.

Moreover, as long as the trajectory stays in $\mathcal R$ we have 
$u=\frac{8}{27}x^3\mu^3\chi^{-3}\le\frac{8}{27}x_b^3\chi_{\mathrm c}^{-3}$, so 
that $t=(1-u)/\mu$ is as close to $1$ as we wish, provided $x_b$ is small 
enough. Hence, by \eqref{eq:LamB}, $B(x,t)\ge\frac94na\,x^{\,n-3}(1-\epsilon)$, 
and choosing $x_b$ small we may assume
\begin{equation}\label{eq:Bcond}
B(x,t)>g(\chi_{\mathrm c})\qquad\text{for every }x\in\,]0,x_b] 
\text{ and every }(\chi,y)\in\mathcal R .
\end{equation}
For $n=3$ this is the naked condition $a>\frac{2(26+15\sqrt3)}{27}$; 
for $n=1,2$ it is automatic, since $x\mapsto x^{\,n-3}$ is decreasing, so that 
the infimum of $B$ over $]0,x_b]$ is attained at $x_b$ and diverges as $x_b\to0^+$.

The two remaining sides of $\partial\mathcal R$ are left to check:

\emph{$\bullet$ $y=y_*$.} Since $\mu\,t=1-u\le1$, \eqref{eq:weq} gives
$$
\frac{\mathrm dw}{\mathrm ds}\le\frac{1+\Lambda}{y}-\frac32\,\Lambda\,\chi ,
$$
so that $\mathrm dw/\mathrm ds\le0$ as soon as 
$\frac32\,B\chi^4y\ge1+B\chi^3$, i.e.\ as soon as 
$B\ge\frac{2}{\chi^3(3\chi y-2)}$. For $y=1$ the right hand side is $h(\chi)$, 
which is decreasing on $]2/3,1[$, so that \eqref{eq:Bcond} and 
\eqref{eq:chic} give $B>h(\chi_{\mathrm c})\ge h(\chi)$ for every 
$\chi\ge\chi_{\mathrm c}$, with strict inequality. By continuity, the same 
holds with $y$ replaced by any $y_*$ sufficiently close to $1$. Hence 
$w$ is non increasing on this side, i.e.\ $y$ increases along the backward 
evolution and the trajectory cannot cross it.

\emph{$\bullet$ $\chi=\chi_{\mathrm c}$.} Since $\mu'<0$, \eqref{eq:chieq} gives
$$
\frac{\mathrm d\chi}{\mathrm ds}\le
\chi\left[1-\frac\Lambda2+\frac{\chi(1+\Lambda)}{2y}\right],
$$
and for $y=1$ the bracket equals 
$\frac{\chi^3(1-\chi)}{2}\left[g(\chi)-B\right]$, which is negative at 
$\chi=\chi_{\mathrm c}$ by \eqref{eq:Bcond}. Again by continuity the same holds 
for $y\ge y_*$, with $y_*$ close enough to $1$. Hence $\chi$ increases along 
the backward evolution on this side, and the trajectory cannot cross it either.

Summing up, for $y_*<1$ close enough to $1$ and $x_b$ small enough, 
$\mathcal R$ is invariant under backward evolution. Since $(\chi,y)$ remains in 
the bounded region $\mathcal R$, where $u>0$ and therefore $r\,r_x>0$, the 
solution is defined for every $x\in\,]0,x_b]$, with $y\ge y_*>0$: in particular 
$y$ never vanishes, and $x(t)\to 0$. Finally $\chi\ge\chi_{\mathrm c}$ together with 
\eqref{eq:chi2} gives $2m/r\ge\chi_{\mathrm c}^2=4-2\sqrt3$, hence 
$r\le 2m(x)/\chi_{\mathrm c}^2\to0$ as $x\to0$, so that $t\to t_s(0)=1$ and the 
trajectory reaches the central singularity.
\end{proof}

\begin{remark}\label{rem:trapping}
The argument is not sharp: it only produces \emph{some} threshold $y_*<1$, 
while numerical integration of \eqref{eq:chieq}--\eqref{eq:weq} shows the set 
of captured initial data to be considerably larger. What matters is that this 
set is non empty and open, which suffices to destroy uniqueness. The bound 
$2m/r\ge4-2\sqrt3$ also localizes the whole family between the apparent horizon 
and $r=(2+\sqrt3)\,m\simeq3.73\,m$, slightly outside the photon sphere value 
$r=3m$. Although Proposition \ref{prop:naked} is stated for $y_0<1$, the same 
invariant-region argument applies unchanged to $y_0=1$, recovering the fact 
that the null extension reaches the central singularity -- as we will see in Section \ref{sec:ps-ns}, Theorem 
\ref{thm:end} proves the same  by a 
more classical technique. 
\end{remark}


\begin{theorem}\label{thm:ext}    
The marginally bound spherical dust collapse admits extensions of the photon 
surface $r=3M$ in the interior of the cloud, and their structure is dictated by 
the nature of the central singularity.
\begin{enumerate}
    \item If the central singularity is covered, the extension is unique. It is 
    the null hypersurface
    $$
    \{(t,x(t),\theta,\phi)\,:\,t<t_0\},
    $$
    with the property that, fixed $\theta$ and $\phi$, the corresponding curve 
    is a null radial geodesic, i.e.\ the couple $(x(t),1)$ solves system 
    \eqref{eq:ODE1a-mb}--\eqref{eq:ODE1b-mb} with $k(x)\equiv 0$, given initial 
    data $x(t_0)=x_b$ and $y(t_0)=1$.
    \item If the central singularity is naked, the above null hypersurface is 
    still an admissible extension, but it is no longer the only one: by 
    Proposition \ref{prop:naked}, every $y_0\in[y_*,1[$ gives rise to a further, 
    timelike, admissible extension, which also reaches the central singularity.
\end{enumerate}
In particular, the uniqueness of the extension of the photon surface 
characterizes the black hole case.
\end{theorem}

\begin{remark}\label{rem:Cao}
The initial condition $x(t_0)=x_b$ ensures that $m$ is $C^0$ along the boundary of the star. Naturally, the Misner-Sharp mass remains constant on the exterior solution, and consequently, on the exterior photon surface at $r=3M$. Therefore, demanding that $m$ possesses $C^1$ regularity at the star's boundary effectively means requiring 
$$\frac{\mathrm d}{\mathrm dt}m(x(t_0))=0,$$
which is claimed in \cite{Cao:2019vlu} to be assumed \footnote{Observe that in \cite{Cao:2019vlu} the Misner-Sharp mass is denoted by $E$.} as a second initial condition, in addition to $x(t_0)=x_b$.  
But actually, that would give $\dot x(t_0)=0$ since the mass depends only on $x$; moreover from the second order ODE \eqref{eq:ps-dust} we have that when $\dot x(t)=0$ then $\ddot x(t)=1/(r\,r_x)>0$ and so a critical point of the curve $x(t)$ is necessarily a strict local minimum. That excludes this possibility because, if we search for an extension of the photon surface in the interior portion of the model, we would rather look for a function $x(t)<x_b$, and this cannot happen in this case where,  in a neighborhood of the initial time $t_0$, $x(t)$ would attain strictly greater values than $x(t_0)=x_b$.

The apparent contradiction is indeed resolved because, as a matter of fact, the 
condition used in \cite{Cao:2019vlu} is precisely the one derived from our 
Theorem \ref{thm:ext}, as one can verify from \cite[(5.41)]{Cao:2019vlu}, 
although it is misinterpreted there, as explained above. Let us stress that the 
comparison concerns the null extension, which by Theorem \ref{thm:ext} is the 
only admissible one when the central singularity is covered.
\end{remark}

\subsection{Photon surface analysis}\label{sec:ps-ns}
Let us now study the qualitative behavior of the extensions found in Theorem 
\ref{thm:ext}, in order to understand whether they cover the singularity 
developing in the LTB model. By Theorem \ref{thm:ext}, when the central 
singularity is covered the extension is unique and null, so that all the 
information is carried by the outgoing radial null geodesic equation 
\eqref{eq:nullgeo}; when it is naked, the same equation describes the null 
extension, the timelike ones having already been discussed in Proposition 
\ref{prop:naked}.

For this analysis, it will be more useful to denote the null extension of the 
photon surface as a solution of the outgoing null radial geodesic expressing 
$t$ as a function of $x$, i.e.
\begin{equation}\label{eq:nullradial}
\frac{\mathrm dt(x)}{\mathrm dx}= {r_x(t(x),x)},\qquad t(x_b)=t_0,  
\end{equation}
that is, equation \eqref{eq:nullgeo} supplemented with the initial datum at the 
star boundary, and defined in a left neighborhood of $x_b$, possibly extendable 
up to the regular centre $x=0$. This is possible because $r_x>0$ and therefore 
$x(t)$ can be inverted. Consequently, we can exploit all the techniques known 
in the literature and based on comparison theorems in ODE 
\cite{Mena:2001dr,Giambo:2002tp}, together with the subsolution property 
\eqref{eq:subsol} of the apparent horizon curve recalled in Section 
\ref{sec:endstate}.

Recalling that $t_h(x_b)>t_0$, we have that  $t(x)<t_h(x)$, $\forall x<x_b$. Please observe that there still might be the case that $t(0)=t_h(0)=1$, because on the singular center of symmetry $(r=0,t=1)$ the regularity needed to apply comparison theorems in ODE is lost. 

Then, the null extension of the photon surface will cover the singularity 
\textit{completely} (i.e. including the center) if one of the two following 
instances occurs:
\begin{enumerate}
    \item $\lim_{x\to 0^+}t(x)<t_s(0)=1$, or
    \item $\lim_{x\to 0^+}t(x)=t_s(0)=1$ but there are no null radial outgoing geodesics $t_g(x)$ such that $t_g(0)=1$ and $t_g(x)<t(x)$, $\forall x\in]0,x_b]$.  
\end{enumerate}

We will observe that the situation varies depending on whether the central singularity is naked or not. The initial scenario arises when, in \eqref{eq:mu}, we have $n=1,2$, or $n=3$ and $a>{2(26+15\sqrt 3)}/{27}$. These cases specifically align with those where we can identify a super-solution for \eqref{eq:nullgeo} of the form $t_\xi=1+\xi x^n$, selecting $\xi$ such that $t_\xi(x)<t_h(x)$ in a right neighborhood of $x=0$. The star boundary $x_b$ will be chosen in this neighborhood.

Assume that it is possible to select $t_0$ as defined in \eqref{eq:t0} such that $t_\xi(x_b) < t_0 < t_h(x_b)$. In this case, the continuation of the solution to the Cauchy problem \eqref{eq:nullradial} extends to smaller values of $x$, reaching $x = 0$, which results in $t(x) \to 1$. Consequently, the extension of the photon surface originates from the central singularity. Moreover, consider a time $t^*\in]t_\xi(x_b),t_0[$, the radial null geodesic with initial condition $t(x_b)=t^*$ will be extended backward up to $x=0$, again with $t(x)\to 1$, giving rise to a photon (actually, an infinite number of photons, given the arbitrariness of $t^*$) emanating from the central singularity but escaping the photon surface. 

In view of the above, and recalling that 
$$
r(t_h(x_b),x_b)=2m(x_b)<3m(x_b)=r(t_0,x_b),
$$
we only have to check that $t_\xi(x_b)<t_0$. Using \eqref{eq:t0} we obtain that 
this condition is satisfied in a right neighborhood of $x=0$ provided $n=1,2$, 
or $n=3$ and $a>(2/3)^{3/2}$, which is implied by the naked singularity 
condition recalled in Section \ref{sec:endstate}, since 
$(2/3)^{3/2}<\frac{2(26+15\sqrt3)}{27}$. Therefore, the mass profiles ensuring  a central naked singularity also ensure that $t_\xi(x_b)<t_0$, and the null extension of the photon surface extends up to the central singularity.

The situation described here is represented in Figures \ref{fig:2} and \ref{fig:3}, where the parameters are set to produce a naked central singularity. 
\begin{figure}[h]
  \centering
  \includegraphics[width=0.9\linewidth]{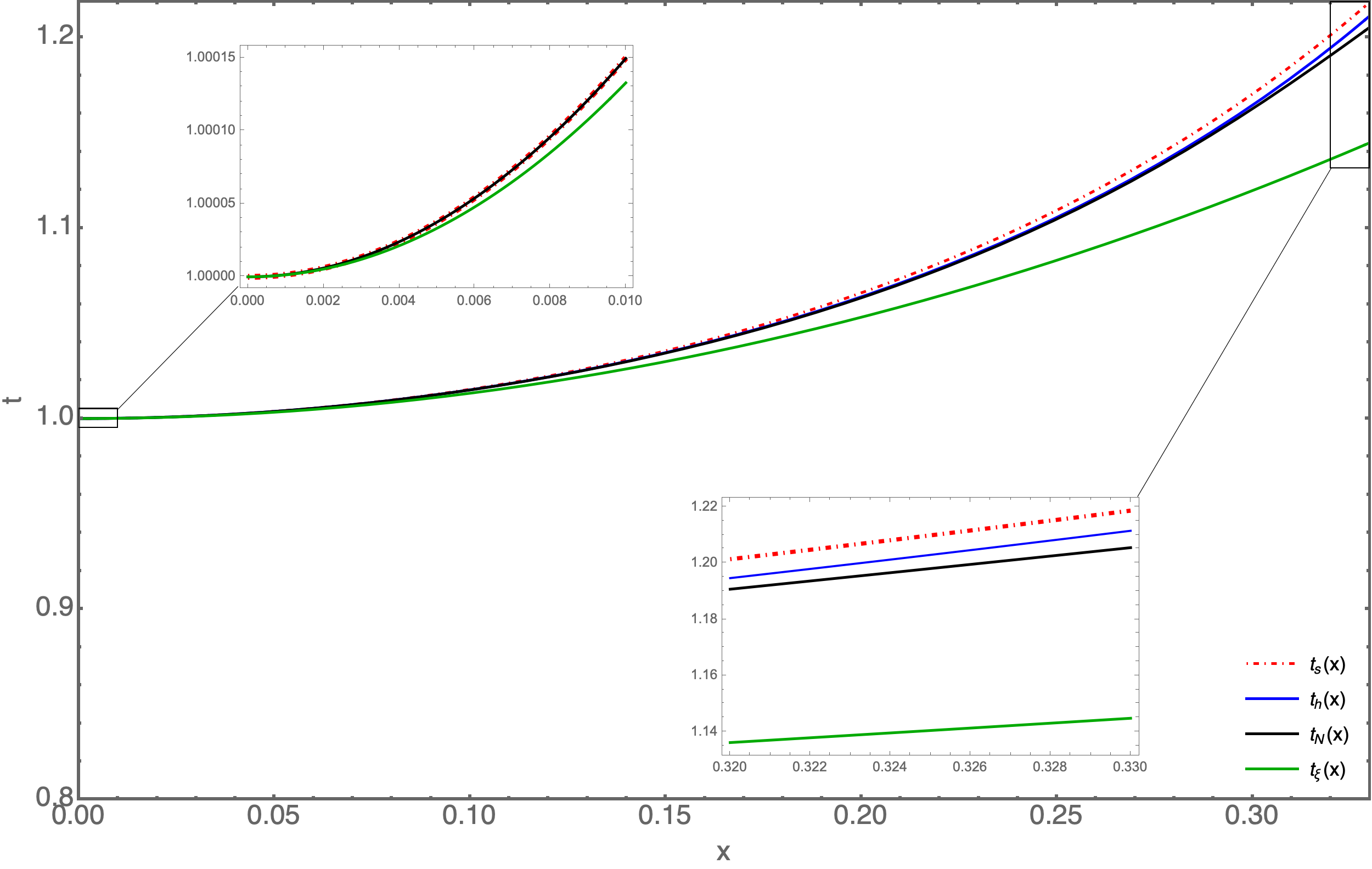}
  \caption{The null extension $t_N(x)$ of the photon surface  in case when the central singularity is naked. The parameters of \eqref{eq:mu} are set as follows: $n=2$, $a=1.5$, with the star boundary $x_b=0.33$. The curve $t_\xi(x)$, with $\xi=1.35$, is a supersolution of the ODE \eqref{eq:nullgeo} in $[0,x_b]$. The curve $t_N(x)$ is extended back to the central singularity.}
  \label{fig:2}
\end{figure}
\begin{figure}[h]
  \centering
  \includegraphics[width=0.9\textwidth]{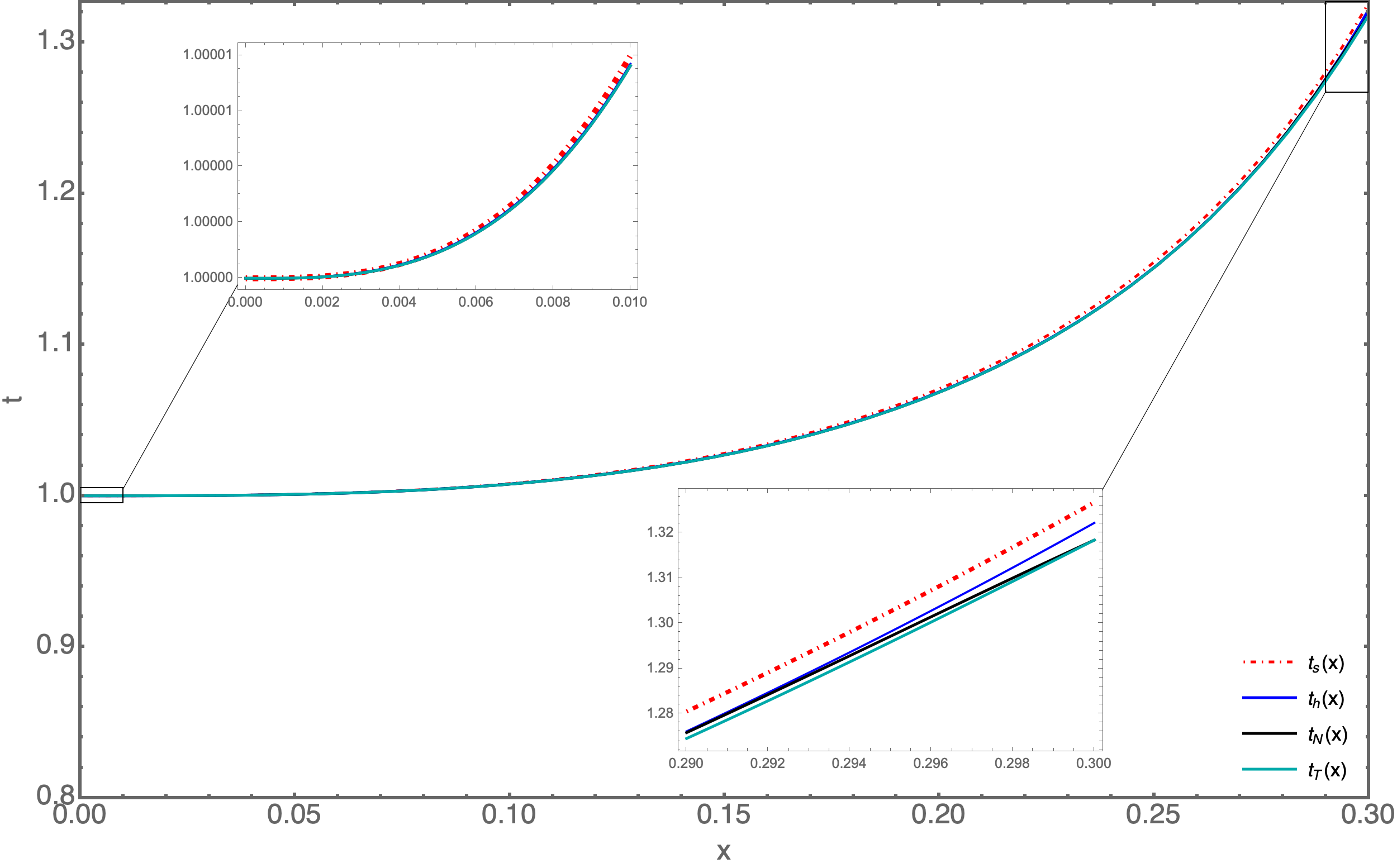}
  \caption{Extensions of the photon sphere with a central naked 
singularity. Here $n=3$, $a=8$ and $x_b=0.3$. 
The null extension $t_N(x)$ (in black) and a timelike one $t_T(x)$ ($y_0=0.9$, in cyan) both reach the central singularity: uniqueness fails.} 
\label{fig:3}
\end{figure}

The curve $t_\xi(x)$ is a supersolution of \eqref{eq:nullgeo} and $x_b$ is chosen such that $t_0\in]t_\xi(x_b),t_h(x_b)[$. The solution of \eqref{eq:nullradial} can be extended back for $x<x_b$: it cannot cross  the subsolution $t_h(x)$ from below nor the supersolution $t_\xi(x)$ from above and then it must extend up to the central singularity $(x=0,t=1)$.

Conversely, when $n>3$ or $n=3$ with $a$ sufficiently small, resulting in a covered central singularity, the photon surface will extend back to the regular center, see Figure \ref{fig:4}. 
\begin{figure}[h]
  \centering
  \includegraphics[width=0.9\textwidth]{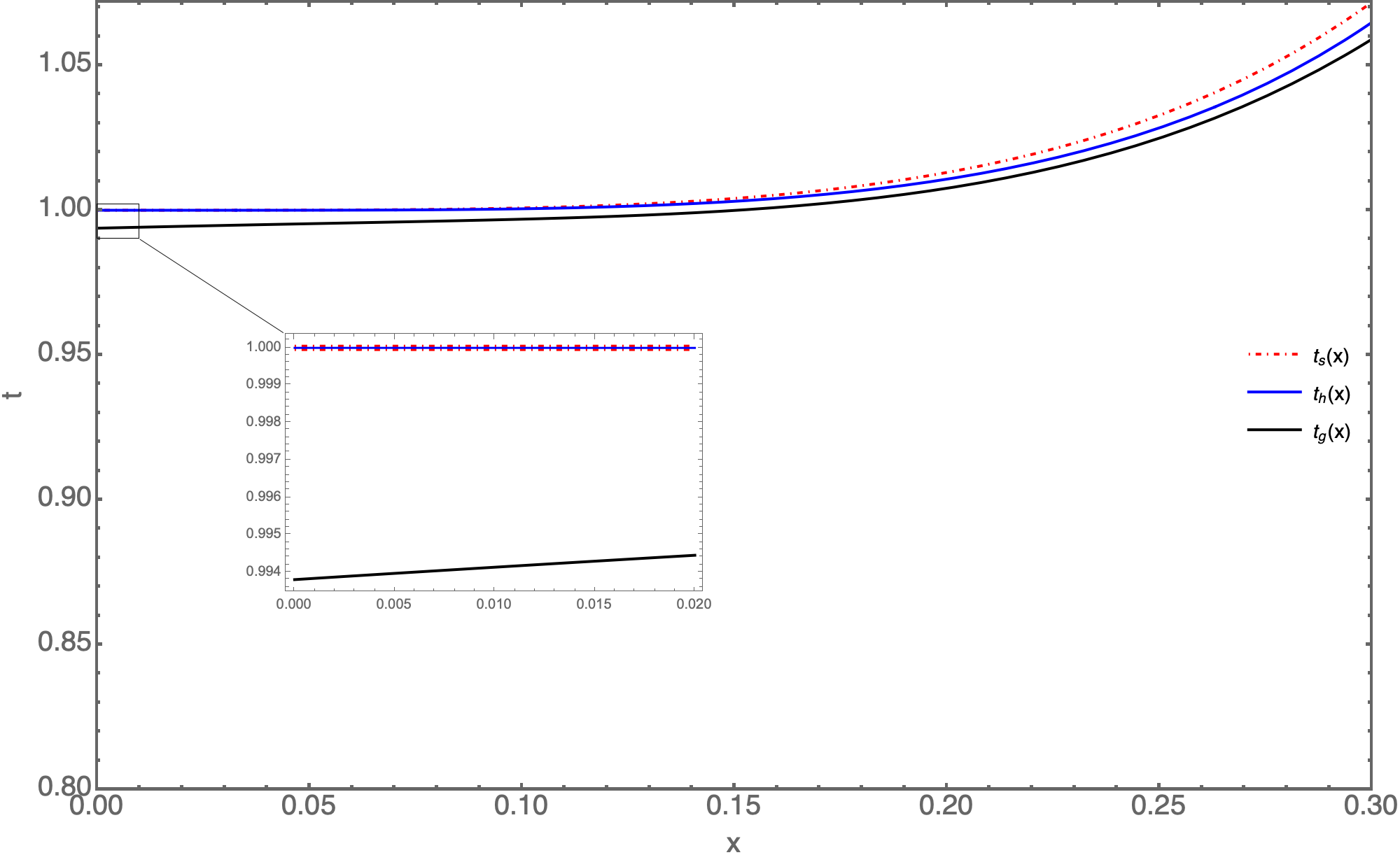}
  \caption{The extension of the photon sphere, in case the central singularity is not naked, is necessarily  continued back to the \textit{regular} centre $t<1$. Here $n=4$, $a=8$, $x_b=0.3$.}
  \label{fig:4}
\end{figure}
If it were to reach the central singularity, it would represent a null radial geodesic escaping the central singularity, violating the condition that the central singularity itself is trapped in this context.

We can summarize the above discussion as follows. 

\begin{theorem}\label{thm:end}
In the marginally bound spherical dust collapse, the null extension of the 
photon surface established in Theorem \ref{thm:ext} extends back to the central 
singularity if and only if the latter is naked. In this case, there exist null 
radial geodesics $t_g(x)$ emanating from the central naked singularity such 
that $t_g(x)<t(x)$, $\forall x\in ]0,x_b]$, where $t(x)$ is the solution of 
\eqref{eq:nullradial}: the photon surface, although reaching the singularity, 
is therefore insufficient to cover it. By Proposition \ref{prop:naked}, in this 
case the timelike extensions reach the central singularity as well.

\noindent Otherwise, if the central singularity is covered, the extension is 
unique by Theorem \ref{thm:ext}, and it extends to the regular centre.
\end{theorem}

\section{Discussion and conclusions}\label{sec:outro}
The study of photon surfaces has generated increasing attention in recent years, particularly due to their fundamental role in the causal structure and observational features of spacetimes. 
In this work, we investigated the behavior and possible extensions of photon surfaces within a dynamically evolving spherically symmetric spacetime, particularly in the context of the marginally bound Lemaitre-Tolman-Bondi (LTB) dust collapse model. Building on the general framework for photon surfaces established in earlier literature, we derived the dynamical equations governing their evolution and showed that, under suitable initial conditions, these surfaces can be continued into the 
interior of the collapsing matter. The explicit form of the metric coefficient 
in the marginally bound case allowed us to recast the problem as an exact planar 
dynamical system, whose analysis shows that such a continuation is unique -- and 
then null -- if and only if the central singularity is covered, whereas in the 
naked case a one-parameter family of timelike extensions is available as well.
This analysis not only confirms but also refines previous claims, clarifying 
that the condition used in earlier works - interpreted there as a requirement for 
the continuity of the mass profile - is more accurately understood in terms of the 
geometric nature of null hypersurfaces.

Our results indicate that, when the central singularity is covered, the only 
physically consistent extension of the photon surface $r = 3M$ into the 
collapsing cloud is given by a surface generated by outgoing null radial 
geodesics. This extension satisfies the photon surface equation and adheres to the junction conditions at the boundary with the Schwarzschild exterior. That is a static geometry, and this is precisely why the photon sphere $r=3M$ is timelike: the light cones at fixed r do not change with time, so circular null orbits can remain tangent to the timelike worldtube $r=3M$. Inside the LTB cloud, the metric coefficients depend on comoving time, and the light cones tilt progressively inward as the dust collapses. A timelike hypersurface playing the role of a photon surface would need to track these changing light cones without letting tangent photons shear away from it, which becomes impossible once the collapse is sufficiently homogeneous near the centre (see Lemma \ref{lem:ext}). The photon surface, therefore, has to become null: in the time-dependent 
collapse here considered, letting the surface coincide with a congruence of 
radial null geodesics is the only way to preserve the trapping property -- as 
long as the collapse ends in a black hole. If instead the singularity is naked, 
the inhomogeneity of the collapse near the centre -- measured by $\Lambda$ in 
\eqref{eq:weq} -- becomes strong enough to balance this effect, and timelike 
photon surfaces survive alongside the null one.

Importantly, the nature of the central singularity - whether it is covered or 
naked - plays a crucial role in determining the global structure of the photon surface. 
When the singularity is covered, the photon surface terminates at the regular center before the singularity forms. Conversely, if the singularity is naked, the photon surface can extend all the way back to the central singularity itself. 
However, in such cases, the surface is not sufficient to prevent the escape of photons from the singularity, highlighting its limited covering capability. Previous studies have established that naked singularities create different lensing effects compared to black holes \cite{Virbhadra:2022iiy}. Additionally, research by Ortiz et al. \cite{Ortiz:2015rma} demonstrates differences in the shadow evolution between naked singularities and black holes.

It is worth stressing that the mechanism behind this dichotomy is of a dynamical 
nature: by \eqref{eq:LamB}, the competition between the two terms of 
\eqref{eq:weq} is governed by $B$ alone, and changes character exactly at 
$\chi_{\mathrm c}=\sqrt3-1$ -- the same value which, translated into $a_{\mathrm 
c}$, separates black hole and naked singularity formation. Loss of uniqueness of 
the photon surface and loss of censorship of the central singularity are thus 
two readings of the same bifurcation, and it would be interesting to understand 
how far this correspondence extends beyond the model considered here.


While both objects eventually appear similar to distant observers, their formation dynamics differ significantly, also due to the infinitely redshifted photons that escape from the central naked singularity to the external region -- a phenomenon documented in earlier work \cite{Dwivedi:1998ts,Giambo:2006vr}.  
Our findings in the present paper support these observations by showing that shadow evolution also varies in terms of photon surface extension, suggesting that understanding photon surface geometry is crucial for distinguishing these compact objects, as the formation and features of the shadow are determined  by the whole underlying spacetime structure.

From an observational standpoint, the key implication of Theorem \ref{thm:end} is that the early-time evolution of the shadow differs between the two end-states. When the central singularity is naked, the photon surface extends to the singularity itself but fails to prevent the escape of null geodesics (those satisfying $t_g(x) < t(x)$); it suggests that the shadow forms with a delay and grows more slowly than in the covered case, where the photon surface reaches the regular centre and the shadow develops monotonically. In principle, a time-resolved sequence of shadow images—or, equivalently, the spectral evolution of photons traversing the collapsing cloud—could capture this distinction. However, it must be remembered that the collapse timescale $\sim GM/c^3$, as one can infer from standard dimensional analysis -- this fact and the opacity of realistic matter severely limit current observational prospects, making such tests a target for next-generation instruments rather than present facilities.

These findings, relying on purely analytical studies, refine previous contributions -- in particular \cite{Cao:2019vlu} -- and contribute to the ongoing discussion about the role of photon surfaces in dynamical spacetimes and their potential connection to cosmic censorship, showing that photon surfaces, though geometric in nature, are deeply intertwined with the causal and global structure of the underlying spacetime. Future research may explore whether similar conclusions hold in more general 
collapse scenarios. The non-marginally bound case is currently under 
investigation \cite{GiamboLucamarini:inprep}: there the metric coefficient is no 
longer explicit, and the exact identities 
\eqref{eq:chieq}--\eqref{eq:weq} exploited here must be replaced by asymptotic 
estimates, which makes the analysis considerably more involved, although the 
dichotomy appears to persist. Models with non-zero pressure, as well as the 
presence of rotation or other deviations from spherical symmetry, remain open.

\subparagraph*{Acknowledgement.}{
The authors acknowledge the support of INdAM. In particular, RG is partially supported by "INdAM--GNAMPA Project" CUP E53C25002010001.
The authors thank  LM Cao and Y Song for their insightful discussions, the anonymous Referees for their helpful remarks and suggestions. The authors are grateful to the Editors for allowing a late revision of the  manuscript.
}


\bibliography{source}
\end{document}